\documentclass[12pt]{article}

\def\UseBibLatex{1}

\makeatletter
\def\input@path{{styles/}}
\makeatother

\providecommand{\BibLatexMode}[1]{}

\renewcommand{\BibLatexMode}[1]{#1}

\ifx\UseBibLatex\undefined%
  \renewcommand{\BibLatexMode}[1]{}
\fi

   \usepackage[bibencoding=utf8,style=alphabetic,backend=biber]{biblatex}%
   \usepackage{sariel_biblatex}%

\usepackage[cm]{fullpage}%
\usepackage{amsmath}%
\usepackage{amssymb}%
\usepackage[table]{xcolor}%

\usepackage{mathcalb}%

\usepackage[T1]{fontenc}

\usepackage[amsmath,thmmarks]{ntheorem}%

\usepackage{titlesec}%
\usepackage{xcolor}%
\usepackage{mleftright}%
\usepackage{xspace}%
\usepackage{graphicx}
\usepackage{hyperref}%
\usepackage[inline]{enumitem}
\usepackage{hyperref}%
\usepackage[ocgcolorlinks]{ocgx2}
\usepackage{euscript}%
\usepackage{todonotes}

\hypersetup{%
      unicode,
      breaklinks,%
      colorlinks=true,%
      urlcolor=[rgb]{0.25,0.0,0.0},%
      linkcolor=[rgb]{0.5,0.0,0.0},%
      citecolor=[rgb]{0,0.2,0.445},%
      filecolor=[rgb]{0,0,0.4},
      anchorcolor=[rgb]={0.0,0.1,0.2}%
}

\titlelabel{\thetitle. }%

\theoremseparator{.}%

\theoremstyle{plain}%
\newtheorem{theorem}{Theorem}[section]

\newtheorem{lemma}[theorem]{Lemma}

{
\theoremseparator{}%
\newtheorem{lemma_no_dot}[theorem]{Lemma}
}

\theoremseparator{.}%

\newtheorem{claim}[theorem]{Claim}%

\newtheorem{observation}[theorem]{Observation}

\theoremstyle{plain}%
\theoremheaderfont{\sf} \theorembodyfont{\upshape}%
\newtheorem*{remark:unnumbered}[theorem]{Remark}%
\newtheorem{remark}[theorem]{Remark}%

\newtheorem{defn}[theorem]{Definition}

\newtheorem{problem}[theorem]{Problem}

\theoremheaderfont{\em}%
\theorembodyfont{\upshape}%
\theoremstyle{nonumberplain}%
\theoremseparator{}%
\theoremsymbol{\myqedsymbol}%
\newtheorem{proof}{Proof:}%

\providecommand{\emphind}[1]{}%
\renewcommand{\emphind}[1]{\emph{#1}\index{#1}}

\definecolor{blue25emph}{rgb}{0, 0, 11}

\providecommand{\emphic}[2]{}
\renewcommand{\emphic}[2]{\textcolor{blue25emph}{%
      \textbf{\emph{#1}}}\index{#2}}

\providecommand{\emphi}[1]{}%
\renewcommand{\emphi}[1]{\emphic{#1}{#1}}

\definecolor{almostblack}{rgb}{0, 0, 0.3}

\providecommand{\emphw}[1]{}%
\renewcommand{\emphw}[1]{{\textcolor{almostblack}{\emph{#1}}}}%

\providecommand{\emphOnly}[1]{}%
\renewcommand{\emphOnly}[1]{\emph{\textcolor{blue25emph}{\textbf{#1}}}}

\newcommand{\myqedsymbol}{\rule{2mm}{2mm}}
\newcommand{\EliotThanks}[1]{%
   \thanks{%
      Department of Computer Science; University of Illinois; 201
      N. Goodwin Avenue; Urbana, IL, 61801, USA; {\tt
         erobson2\atgen{}illinois.edu}; {\tt
         \url{https://eliotwrobson.github.io/}.} #1}}
\newcommand{\SarielThanks}[1]{%
   \thanks{%
      Department of Computer Science; %
      University of Illinois; %
      201 N. Goodwin Avenue; %
      Urbana, IL, 61801, USA; %
      \href{mailto:spam@illinois.edu}{sariel@illinois.edu}; %
      \url{http://sarielhp.org/}.%
   #1%
   }%
}
 \newcommand{\BenThanks}[1]{%
   \thanks{Department of Computer Science;
           University of Texas at Dallas; Richardson, TX 75080, USA;
      {\tt benjamin.raichel\atgen{}utdallas.edu}; {\tt
         \url{http://utdallas.edu/\string~benjamin.raichel}.} #1}}
\newcommand{\atgen}{\symbol{'100}}

\newcommand{\HLink}[2]{\hyperref[#2]{#1~\ref*{#2}}}
\newcommand{\HLinkSuffix}[3]{\hyperref[#2]{#1\ref*{#2}{#3}}}

\newcommand{\figlab}[1]{\label{fig:#1}}
\newcommand{\figref}[1]{\HLink{Figure}{fig:#1}}

\newcommand{\thmlab}[1]{{\label{theo:#1}}}
\newcommand{\thmref}[1]{\HLink{Theorem}{theo:#1}}

\newcommand{\lemlab}[1]{\label{lemma:#1}}
\newcommand{\lemref}[1]{\HLink{Lemma}{lemma:#1}}%

\providecommand{\deflab}[1]{\label{def:#1}}

\newcommand{\defrefregY}[2]{\hyperref[def:#1]{{\textcolor{yellow}{#2}}}}

\definecolor{blackish}{rgb}{0.14, 0, 0.0}
\newcommand{\defrefY}[2]{%
   \textcolor{blackish}{%
      \renewcommand\color[2][]{}%
      \defrefregY{#1}{#2}%
   }%
}

\providecommand{\eqlab}[1]{}%
\renewcommand{\eqlab}[1]{\label{equation:#1}}
\newcommand{\Eqref}[1]{\HLinkSuffix{Eq.~(}{equation:#1}{)}}

\providecommand{\remove}[1]{}%
\newcommand{\Set}[2]{\left\{ #1 \;\middle\vert\; #2 \right\}}

\newcommand{\pth}[1]{\mleft(#1\mright)}%

\newcommand{\Vertices}{\Mh{\mathsf{V}}}%
\newcommand{\VV}{\Vertices}%

\newcommand{\ceil}[1]{\mleft\lceil {#1} \mright\rceil}

\newcommand{\brc}[1]{\left\{ {#1} \right\}}

\newcommand{\cardin}[1]{\left\lvert {#1} \right\rvert}%

\renewcommand{\th}{th\xspace}
\newcommand{\ds}{\displaystyle}%

\renewcommand{\Re}{\mathbb{R}}%
\newcommand{\ZZ}{\mathbb{Z}}%

\newlist{compactenumA}{enumerate}{5}%
\setlist[compactenumA]{topsep=0pt,itemsep=-1ex,partopsep=1ex,parsep=1ex,%
   label=(\Alph*)}%

\newlist{compactenuma}{enumerate}{5}%
\setlist[compactenuma]{topsep=0pt,itemsep=-1ex,partopsep=1ex,parsep=1ex,%
   label=(\alph*)}%

\newlist{compactenumI}{enumerate}{5}%
\setlist[compactenumI]{topsep=0pt,itemsep=-1ex,partopsep=1ex,parsep=1ex,%
   label=(\Roman*)}%

\newlist{compactenumi}{enumerate}{5}%
\setlist[compactenumi]{topsep=0pt,itemsep=-1ex,partopsep=1ex,parsep=1ex,%
   label=(\roman*)}%

\newlist{compactitem}{itemize}{5}%
\setlist[compactitem]{topsep=0pt,itemsep=-1ex,partopsep=1ex,parsep=1ex,%
   label=\ensuremath{\bullet}}%

\numberwithin{figure}{section}%
\numberwithin{table}{section}%
\numberwithin{equation}{section}%

\usepackage{stmaryrd}%

\newcommand{\Term}[1]{\textsf{#1}}%
\newcommand{\WSPD}{\Term{WSPD}\xspace}

\newcommand{\repX}[1]{\zeta^{}_{#1}}%

\DefineNamedColor{named}{AlgorithmColor}{cmyk}{0.07,0.90,0,0.34}

\newcommand{\AlgorithmI}[1]{{%
      \textcolor[named]{AlgorithmColor}{\texttt{\bf{#1}}}%
   }}

\newcommand{\concat}{\odot}%
\newcommand{\mX}[1]{\mathsf{m}\pth{#1}}%
\newcommand{\algDistill}{\AlgorithmI{distill}\xspace}

\newcommand{\crA}{\curve_A}
\newcommand{\crB}{\curve_B}

\newcommand{\craA}{\curveA_A}%
\newcommand{\craB}{\curveA_B}

\newcommand{\SaveContent}[2]{%
   \expandafter\newcommand{#1}{#2}%
}

\providecommand{\Mh}[1]{{#1}}%

\newcommand{\curve}{\Mh{\pi}}

\newcommand{\curveA}{\Mh{\sigma}}
\newcommand{\curveB}{\Mh{\xi}}
\newcommand{\curveC}{\Mh{\eta}}

\newcommand{\dY}[2]{\left\| #1 - #2 \right\|}

\newcommand{\eps}{\varepsilon}

\newcommand{\FMS}{\EuScript{X}\index{metric space}}
\newcommand{\FMSB}{\EuScript{Y}\index{metric space}}

\newcommand{\dC}{\mathcalb{d}}%
\newcommand{\DistChar}{\dC}%
\newcommand{\dmY}[2]{\DistChar\pth{#1,#2}}%
\newcommand{\dG}{\DistChar_\G}%
\newcommand{\dGY}[2]{\dG\pth{#1,#2}}%
\newcommand{\ts}{\hspace{0.6pt}}

\newcommand{\DistCharB}{\Mh{\mathrm{d}'}}%
\newcommand{\dmBY}[2]{\DistCharB\pth{#1,#2}}%
\newcommand{\Ground}{{U}}
\newcommand{\GroundB}{{Y}}

\newcommand{\LipC}{\mathcal{L}}
\newcommand{\LipX}[1]{\mathcal{L}\pth{#1}}
\newcommand{\constA}{\mathpzc{c}}
\DeclareMathAlphabet{\mathpzc}{OT1}{pzc}{m}{it}

\newcommand{\Distor}{\mathrm{distortion}}

\newcommand{\VR}{\mathcal{V}}
\newcommand{\WR}{\mathcal{W}}
\newcommand{\PB}{{Q}}%
\newcommand{\PC}{{R}}%

\newcommand{\diamC}{\nabla}%
\newcommand{\diamX}[1]{\diamC\pth{#1}}
\newcommand{\diamY}[2]{\diamC^{}_{\!#1}\pth{#2}}

\newcommand{\diamMX}[1]{\diamY{\DistChar}{#1}}

\newcommand{\G}{\mathsf{G}}%
\renewcommand{\H}{\mathsf{H}}%
\newcommand{\EG}{\mathsf{E}}%

\providecommand{\Edges}{\Mh{\mathsf{E}}}%
\providecommand{\EE}{\Edges}%
\providecommand{\EdgesX}[1]{\Edges\pth{#1}}%
\providecommand{\EGX}[1]{\EdgesX{#1}}%

\newcommand{\Tree}{\mathcal{T}}

\newcommand{\deY}[2]{\DistChar_2\pth{#1,#2}}

\newcommand{\seclab}[1]{\label{sec:#1}}
\newcommand{\secref}[1]{\HLink{Section}{sec:#1}}

\providecommand{\P}{\mathsf{P}}%
\renewcommand{\P}{\mathsf{P}}%

\providecommand{\Q}{\mathsf{P}}%
\renewcommand{\Q}{\mathsf{Q}}%

\newcommand{\lenX}[1]{\mleft\| #1 \mright\|}
\newcommand{\GridC}{\mathcal{G}}%
\newcommand{\GridX}[1]{\GridC\pth{#1}}%

\newcommand{\cellA}{\square}%
\newcommand{\nnY}[2]{\mathsf{nn}^{}_{#1}\pth{#2}}

\newcommand{\vorZ}[3]{{#1}^{}_{\!#2}\pth{#3}}
\newcommand{\GridS}{\GridX{\tfrac{1}{8}}}
\newcommand{\shadowX}[1]{\mathcal{S}\pth{#1}}%

\newcommand{\Packing}{\mathcal{N}}%

\newcommand{\clmlab}[1]{\label{claim:#1}}
\newcommand{\clmref}[1]{\HLink{Claim}{claim:#1}}
\providecommand{\etal}{et~al.\xspace}
\renewcommand{\etal}{et~al.\xspace}
\newcommand{\pair}{\mathcalb{p}}%

\newcommand{\WRopt}{\WR_{\mathrm{opt}}}%
\newcommand{\ballC}{\mathcalb{b}}%
\newcommand{\ballY}[2]{\ballC\pth{#1,#2}}%

\newcommand{\riX}[1]{2^{#1}}

\newcommand{\hide}[1]{}
\bibliography{wspd_revisit}%

\title{Well-Separated Pairs Decomposition Revisited}

\author{%
   Sariel Har-Peled%
   \SarielThanks{Work on this paper was partially supported by NSF AF
      award CCF-2317241.}%
   \and%
   Benjamin Raichel%
   \BenThanks{Work on this paper was partially supported by NSF CCF
      award 2311179.}%
   \and%
   Eliot W. Robson%
   \EliotThanks{}%
}%

\date{\today}

\begin{document}

\maketitle

\begin{abstract}
    We revisit the notion of \WSPD (i.e., \emph{well-separated pairs-decomposition}),  presenting a new construction of \WSPD for any finite metric space, and show that it is asymptotically instance-optimal in size. Next, we describe a new \WSPD construction for the weighted unit-distance metric in the plane, and show a bound $O( \eps^{-2} n \log n)$ on its size, improving by a factor of $1/\eps^2$ over previous work.  The new construction is arguably simpler and more elegant.

    We point out that using \WSPD, one can approximate, in near-linear time, the distortion of a bijection between two point sets in low dimensions. As a new application of \WSPD, we show how to shortcut a polygonal curve such that its dilation is below a prespecified quantity. In particular, we show a near-linear time algorithm for computing a simple subcurve for a given polygonal curve in the plane so that the new subcurve has no self-intersection.
\end{abstract}

\section{Introduction}

In 1995, Callahan and Kosaraju \cite{ck-dmpsa-95} published a surprising result showing that the Euclidean metric, for a set $\P$ of $n$ points in $\Re^d$ can be compactly described as the union of $O(n/\eps^d)$ bicliques, where all the distances of edges in a single biclique are the same up to a factor of $1\pm \eps$. Elegantly, the biclique cover is computed in $O(n \log n + n/\eps^d)$ time, with each biclique being represented as a pair of nodes in a constructed tree over the point set, see  \cite{h-gaa-11} for more details. This representation of the distances of $\P$ is a \emphw{$\tfrac{1}{\eps}$-\WSPD} of $\P$.

Having a linear or near-linear size \WSPD is desirable as it readily leads to efficient spanners, distance oracle constructions, and fast all nearest-neighbor computation, among other applications. Thus, there was interest in showing the existence of \WSPD for other metrics:
\begin{compactenumI}
    \smallskip%
    \item Gao and Zhang \cite{gz-wspdu-05} showed how to construct $\tfrac{1}{\eps}$-\WSPD, of size $O(\eps^{-4} n \log n)$, for the unit-distance graph for a set $\P$ of $n$ points in the plane.  They also show a bound of $O_\eps(n^{2-2/d})$ for $d>2$ (they do not state the exact dependency on $\eps$).

    \smallskip%
    \item Har-Peled and Mendel \cite{hm-fcnld-06} showed the existence of such \WSPD for doubling metrics. These metrics are an abstraction of low-dimensional Euclidean space, and the bounds match the Euclidean bounds of Callahan and Kosaraju.

    \smallskip
    \item Gudmundsson and Wong \cite{gw-wspdl-24} showed a $\tfrac{1}{\eps}$-\WSPD of size $O( \lambda^2 \eps^{-4} n \log n))$ for a $\lambda$-low-density graph. Such a graph is a straight-line embedding in the plane, where no disk intersects more than $\lambda$ edges with length exceeding the disk's diameter.

    \smallskip%
    \item Deryckere \etal \cite{dgrsw-wsstc-25} showed the existence of a linear size \WSPD for $c$-packed graphs.
\end{compactenumI}

\subsection*{Our results}

\begin{compactenumI}
    \smallskip%
    \item \textsf{Instance optimal \WSPD.} %
    We show how to construct \WSPD using packings. The new construction has asymptotically optimal-sized \WSPD for the given instance.  See \secref{wspd_via_packings}.

    \smallskip%
    \item \textsf{Better \WSPD for unit disk-graphs.} %
    We revisit the result of Gao and Zhang \cite{gz-wspdu-05} and slightly improve their bound by a factor of $1/\eps^2$ to $O( \eps^{-2} n \log n)$. Our construction is arguably simpler than the construction of Gao and Zhang. Note that there is still a gap of $O(\log n)$ between the upper and lower bounds on the size of the optimal \WSPD in this case. Whatever the case might be, one can, of course, use the optimal construction mentioned above.  See \secref{wspd_disks} for details.

    \smallskip%
    \item \textsf{Approximating distortion of a mapping between two low-dimensional point sets.}  Consider a bijection $f:\P \rightarrow \P'$ between two sets of points, say $\P \subseteq \Re^d$ and $\P' \subseteq \Re^{d'}$. We point out that \WSPD can be used to approximate the distortion of $f$, in $O(n \log n + n/\eps^{\max(d, d')})$. While this is a relatively easy consequence of known results, (it seems like) it was not observed before. See \secref{distortion} for details.

    \smallskip%
    \item \textsf{Computing a simple subcurve.}  We study the problem of given a simple polygonal curve $\curveA$ in the plane, how to compute a simple subcurve $\curveB \subseteq \cup_{p \in \curveA}p$ that shares the same endpoints as the original curve, without having any self-intersection. We present an elegant divide and conquer algorithm for this task that runs in $O(n \log^2 n)$ time, see
    \secref{simple_subcurve}.

    \smallskip%
    \item \textsf{Approximate shortcutting long detours.}
    Given a polygonal curve $\curveA$ in $\Re^d$, consider the task of shortcutting parts of the curve that have high dilation.  Specifically, given two vertices $p_i, p_j$ of $\curveA$, their \emphw{dilation} is the ratio $\lenX{\curveA[p_i, p_j]} / \dY{p_i}{p_j} > \alpha$, where $\lenX{\curveA[p_i, p_j]}$ denotes the distance between $p_i$ and $p_j$ along $\curveA$. Informally, high dilation indicates that the curve ``wonders around'' when moving from $p_i$ to $p_j$, and then comes back to the vicinity of $p_i$.

    Thus, given a parameter $\alpha > 1$, a \emph {$\alpha$-shortcut} between vertex $p_i$ and $p_j$ of $\curveA$  is allowed if their dilation is $\geq \alpha$ -- that is, we replace the portion of the curve between $p_i$ and $p_j$ by the straight segment $p_ip_j$.
    Consider repeatedly performing such shortcuts until the resulting curve has no such shortcut.

    We show that an algorithm, similar to the one for computing a simple subcurve described above, coupled with \WSPD, leads to a near linear time algorithm that computes a polygonal subcurve, using a subset of the vertices of the original curve, where only shortcuts of dilation larger than $(1-\eps)\alpha$ are used, and no $\alpha$-shortcut remains. See \secref{detour} for details.

\end{compactenumI}

\section{Instance-optimal \WSPD and how to construct it}
\seclab{wspd_via_packings}

\subsection{Preliminaries}

\begin{defn}
    \deflab{metric_space_def}%
    A \emphi{metric space} $\FMS$ is a pair $\FMS = (\Ground, \DistChar )$, where $\Ground$ is the ground set and $\DistChar: \Ground \times \Ground \rightarrow [0, \infty)$ is a \emphi{metric} satisfying the following axioms: (i) $\dmY{x}{y} = 0$ if and only if $x =y$, (ii) $\dmY{x}{y} = \dmY{y}{x}$, and (iii) $\dmY{x}{y} + \dmY{y}{z} \geq \dmY{x}{z}$ (triangle inequality).
\end{defn}

\begin{defn}
    For a set $\P \subseteq \Ground$, its \emphi{diameter} is $\diamY{\dC}{\P} = \max_{x,y \in \P} \dmY{x}{y}$.
\end{defn}

\begin{defn}
    For a point $q \in \Ground$, and a set $\P \subseteq \Ground$, its \emphi{nearest-neighbor} in $\P$, is the point
    \begin{equation*}
        \nnY{\P}{q} = \arg \min_{p \in \P} \dmY{q}{p}.
    \end{equation*}
    The distance between $q$ and its nearest neighbor is denoted by $\dmY{q}{\P} = \min_{p \in \P} \dmY{q}{p}$.
\end{defn}

\begin{defn}
    For any two sets $X,Y \subseteq \Ground$, let $\dmY{X}{Y} = \min_{x \in X, y \in Y} \dmY{x}{y}$.
\end{defn}

\begin{defn}
    \deflab{packing}%
    Consider a metric space $(\Ground, \DistChar)$, and a set $\P \subseteq \Ground$.  A set $\Packing \subseteq \P$ is an \emphi{$r$-packing} for $\P$ if the following hold:
    \begin{compactenumi}
        \smallskip%
        \item \emphw{Covering property}: All the points of $\P$ are within a distance $< r$ from the points of $\Packing$. Formally, for all $p \in \P$, $\dmY{p}{\Packing} <r$.

        \smallskip%
        \item \emphw{Separation property}: For any pair of points $x, y \in \Packing$, we have that $\dmY{x}{y} \geq r$.
    \end{compactenumi}
\end{defn}
One can compute such a packing by repeatedly adding any point in $\P$ at a distance $\geq r$ from the current set till no such point remains. Faster algorithms are known in some cases \cite{hr-nplta-15,ehs-agcds-20}.

\begin{defn}
    For a point $x \in \Ground$, and a radius $r \geq 0$, the \emphi{ball} of radius $r$ centered at $x$ is the set
    \begin{math}
        \ballY{x}{r} = \Set{z \in \Ground}{\dmY{x}{z} \leq r}.
    \end{math}
\end{defn}

\subsection{Background}
\seclab{wspd_def}

For a graph $\G= (\VV, \EE)$, and a set $Y \subseteq \VV$, the \emphi{induced subgraph} of $\G$ by $Y$ is $\G_Y = (Y, \{ uv \in \EE \mid u,v, \in Y\}$.
In the following, assume we are given a metric space $(\Ground,\DistChar)$.

\begin{defn}
    For two sets $B,C \subseteq \Ground$, let
    \begin{math}
        B \otimes C =%
        \Set{ bc}{b \in B, c \in C, b \neq c}.
    \end{math}

\end{defn}

\begin{defn}
    \deflab{pair_decomposition}%
    For a point set $\P \subseteq \Ground$, a \emphi{pair decomposition} of $\P$ is a set of pairs
    \[
        \WR = \brc{\bigl. \brc{A_1,B_1},\ldots,\brc{A_s,B_s}},
    \]
    such that
    \begin{enumerate*}[label=(\Roman*)]
        \item $A_i,B_i\subset \P$ for every $i$,
        \item $A_i \cap B_i = \emptyset$ for every $i$, and
        \item $\bigcup_{i=1}^s A_i \otimes B_i = \binom{\P}{2} = \P \otimes \P$.
    \end{enumerate*}
\end{defn}

\begin{defn}
    \deflab{well_separated}%
    The pair $\{\PB, \PC\}$ is \emphi{$\tfrac{1}{\eps}$-separated} by $\DistChar$ if
    \begin{equation*}
        \max \pth{\bigl. \diamY{\DistChar}{\PB},
           \diamY{\dC}{\PC} } \leq \eps \ts \dmY{\PB}{\PC},
        \qquad\text{where}\qquad%
        \dmY{\PB}{\PC} = \ds \min_{x \in \PB, y \in \PC} \dmY{x}{y}.
    \end{equation*}
\end{defn}

\begin{defn}
    \deflab{WSPD}%
    For a point set $\P$, a \emphOnly{well-separated pair decomposition} of $\P$ with parameter $1/\eps$, denoted by \emphw{$\tfrac{1}{\eps}$-\WSPD{}}, is a pair decomposition
    \begin{math}
        \WR = \brc{\bigl.  \brc{A_1,B_1},\ldots,\brc{A_s,B_s}}
    \end{math}
    of $\P$, such that, for all $i$, the sets $A_i$ and $B_i$ are $\tfrac{1}{\eps}$-separated.
\end{defn}

\begin{theorem}[\cite{ck-dmpsa-95}]
    \thmlab{WSPD}%
    For $1 \geq \eps>0$, and a set $\P$ of $n$ points in $\Re^d$, one can construct, in $O \pth{ n \log n + {n}/{ \eps^{d}}}$ time, an $\tfrac{1}{\eps}$-\WSPD of $\P$ of size $O(n/{ \eps^{d}})$.
\end{theorem}

For a pair $\pair = \{B,C \} \in \WR$, its \emphi{diameter} is $\diamX{\pair} = \diamX{ B \cup C}$.

\subsection{An instance  optimal construction of \WSPD}

The input is a metric space $(\Ground, \DistChar )$, and a set $\P \subseteq \Ground$ of $n$ points, and a parameter $\eps \in (0,1)$. The task at hand is to compute $\tfrac{1}{\eps}$-\WSPD of $\P$.

\begin{defn}
    For a set $\Packing \subseteq \P$, the \emphw{Voronoi cell} of $p \in \Packing$ in $\P$ is
    \begin{equation}
        \vorZ{\P}{\Packing}{p} = \Set{ q \in \P}{\smash{\nnY{q}{\Packing}} = p}.
        \eqlab{v-cell}%
    \end{equation}
    The set $\vorZ{\P}{\Packing}{p}$ is the set of all the points of $\P$ that $p$ is their nearest neighbor in $\Packing$, under
    $\dC$.
\end{defn}

\paragraph*{Construction.}
Let $I$ be the set of all integers $i$ such that there is a pair of points $x,y \in \P$ with $\dmY{x}{y} \in [\tfrac{4}{\eps}\riX{i}, \tfrac{8}{\eps}\riX{i}]$ (note, that $i$ might be negative). Next, compute for all $i \in I$, a $\riX{i}/2$-packing $\Packing_i$ of $\P$, and proceed to compute the set of pairs
\begin{equation*}
    \WR_i
    =%
    \Set{\bigl.\{ \vorZ{\P}{\Packing_i}{x}, \vorZ{\P}{\Packing_i}{y}\}}{
       x,y  \in \Packing_i
       \text{ and }
       \dmY{x}{y} \in [\tfrac{2}{\eps}\riX{i}, \tfrac{16}{\eps}
       \riX{i}] },
\end{equation*}
see \Eqref{v-cell}.  The desired \WSPD is $\WR = \cup_{i \in I} \WR_i$.

\begin{remark}
    The above construction algorithm takes polynomial time. In some cases, one can get near-linear running time (e.g., constant doubling dimension with polynomial spread).
\end{remark}

\begin{theorem}
    \thmlab{WSPD_generic}%
    Let $\FMS = (\Ground,\DistChar)$ be a  metric space,  let $\P \subseteq \Ground$ be a set of $n$ points, and $\eps \in (0,1)$ be a parameter. Then the above algorithm computes (in polynomial time) a $\Bigl.\tfrac{1}{\eps}$-\WSPD $\WR$ for $\FMS$, such that, any $c/\eps$-\WSPD $\WRopt$ of $\FMS$ must be of size $\Omega( \cardin{\WR})$, where $c$ is some sufficiently large constant.
\end{theorem}

\begin{proof}
    Consider any pair of points $x,y \in \P$, and let $\ell = \dmY{x}{y}$. Let $i \in I$ be the index such that $\ell \in [\tfrac{4}{\eps}\riX{i}, \tfrac{8}{\eps} \riX{i}]$.  Let $x' = \nnY{\Packing_i}{x}$ and $y' = \nnY{\Packing_i}{y}$. Observe that $\dmY{x}{x'},\dmY{y}{y'} \leq \riX{i}/2$. Thus, we have
    \begin{equation*}
        \ell - \riX{i} \leq \dmY{x'}{y'} \leq \ell + \riX{i}
        \qquad\implies\qquad%
        \dmY{x'}{y'} \in [\tfrac{3}{\eps}\riX{i}, \tfrac{9}{\eps}\riX{i}].
    \end{equation*}
    Thus $x \in X = \vorZ{\P}{\Packing_i}{x'}$ and $y \in Y = \vorZ{\P}{\Packing_i}{y'}$, and $\{X,Y \} \in \WR_i \subseteq \WR$. Furthermore, we have
    \begin{equation*}
        \dmY{X}{Y}
        \geq
        \ell - 2 \cdot \riX{i}
        \geq
        \tfrac{4}{\eps} \riX{i}
        -2 \cdot \riX{i}
        \geq
        \tfrac{2}{\eps}\riX{i}
        \qquad \text{and}\qquad%
        \diamY{\dC}{X},%
        \diamY{\dC}{Y}%
        \leq
        \riX{i}
        =
        \tfrac{\eps}{2} \cdot \tfrac{2}{\eps} \riX{i}
        <
        \eps
        \dmY{X}{Y}.
    \end{equation*}
    Namely, the points $x,y$ are covered by a pair of $\WR_i$, and the pair is $\tfrac{1}{\eps}$-separated. A similar argumentation shows that all the pairs in $\WR$ are $\tfrac{1}{\eps}$-separated.

    \medskip%

    As for the size optimality. Consider $\WR_i$, for some $i \in I$, any pair $\{B,C\} \in \WR_i$ is induced by a pair of points $b,c \in \Packing_i$, that is
    \begin{equation*}
        B = \vorZ{\P}{\Packing_i}{b}\qquad\text{ and }\qquad
        C = \vorZ{\P}{\Packing_i}{c}.
    \end{equation*}
     A pair $\{X,Y\}$, with $X,Y \subseteq \P$, \emphi{dominates} the pair $\{B,C\}$ if $bc \in X \oplus Y$. We break the pairs of $\WR$ into $7$ sets, as follows
    \begin{equation*}
        \VR_c = \bigcup_{i \in \ZZ} \WR_{7i + c}
        \qquad\text{for}\qquad%
        c \in \{0,\ldots, 6\}.
    \end{equation*}

    We claim that no $33/\eps$-separated pair $\{X,Y\}$ can dominate two pairs of (say) $\VR_0$.  For the sake of contradiction, consider two pairs $\{B_1, C_1\}, \{ B_2, C_2\} \in \VR_0$, that are induced by the pairs of points $\pair_1 = \{x_1, y_1\}$ and $\pair_2 = \{x_2,y_2\}$, respectively,  and assume $x_1, x_2 \in X$ and $y_1, y_2 \in Y$, and
    $\pair_1, \pair_2$ are both dominated by $\{X,Y\}$.
    There are two possibilities: \medskip%
    \begin{compactenumI}
        \item The two pairs belong to the set $\WR_i \subseteq \VR_0$.
        It must be that either $\dmY{x_1}{x_2} \geq \riX{i}/2$ or $\dmY{y_1}{y_2} \geq \riX{i}/2$, since $\Packing_i$ is a $\riX{i}/2$-packing. So, assume $\dmY{x_1}{x_2} \geq \riX{i}/2$. We have that
        \begin{equation*}
            \diamMX{X}
            \geq%
            \frac{\riX{i}}{2}
            =%
            \frac{\eps}{32} \cdot \frac{16  }{\eps } \cdot \riX{i}
            \geq
            \frac{\eps}{32}\dmY{x_1}{y_1}
            >
            \frac{\eps}{33}\dmY{X}{Y}.
        \end{equation*}
        Namely, the pair $X, Y$ is not $33/\eps$-separated. A contradiction.

        \medskip
        \item Consider the case that $\{B_1, C_1\} \in \WR_i$, and $\{B_2, C_2\} \in\WR_{i+t}$, where $t >6$. By construction, we have
        \begin{equation*}
            \dmY{x_1}{y_1}
            \leq
            \tfrac{16}{\eps} \riX{i}
            <
            \tfrac{32}{\eps} \riX{i}
            = %
            \tfrac{1}{4\eps} \riX{i+7}
            \leq
            \tfrac{1}{4}\dmY{x_2}{y_2}.
        \end{equation*}
        By the triangle inequality, we have
        \begin{math}
            \dmY{x_2}{y_2}
            \leq
            \dmY{x_2}{x_1}
            +
            \dmY{x_1}{y_1}
            +
            \dmY{y_1}{y_2}.
        \end{math}
        Implying that
        \begin{equation*}
            \dmY{x_2}{y_2}
            \leq
            \dmY{x_2}{x_1}
            +
            \dmY{x_1}{y_1}
            +
            \dmY{y_1}{y_2}
            \leq
            \tfrac{1}{4}\dmY{x_2}{y_2}
            +
            \dmY{x_1}{x_2}
            +
            \dmY{y_1}{y_2}.
        \end{equation*}
        Leading to  the contradiction
        \begin{equation*}
            \tfrac{\eps}{33} \dmY{x_1}{y_1}
            \geq
            \max\bigl( \diamX{X}, \diamX{Y} \bigr)
            \geq
            \max\bigl( \dmY{x_2}{y_2}, \dmY{y_1}{y_2} \bigr)
            \geq
            \tfrac{3}{8}
            \dmY{x_2}{y_2}
            >
            \dmY{x_1}{y_1},
        \end{equation*}
    \end{compactenumI}
    \medskip%
    Let $\WRopt$ be the optimal $\tfrac{33}{\eps}$-\WSPD of $\P$. Any $\tfrac{33}{\eps}$-\WSPD must dominate all the pairs in $\VR_c$, for $c \in \{0,\ldots,6\}$. But the above implies that each pair of $\WRopt$ dominates only a single pair in $\VR_c$. We conclude that
    $\cardin{\WRopt} \geq \max_{c=0}^6 \cardin{\VR_c} \geq \cardin{\WR}/7$.
\end{proof}

\section{\WSPD for weighted unit-distance graphs}
\seclab{wspd_disks}%

\subsection{Construction}

\begin{defn}
    For a set of points $\P \subseteq \Re^d$ the \emphi{weighted unit-distance graph} $\G =(\P,\EE)$ connects any two points $x,y \in \P$, such that $\dY{x}{y} \leq 1$. The \emphw{weight} of such an edge $xy \in \EE$ is $\dY{x}{y}$. The \emphw{shortest-path metric} $\dGY{x}{y}$ is the length of the shortest-path between $x$ and $y$ in $\G$. If $\G$ is connected, then its diameter is bounded by $n-1$
\end{defn}

\paragraph{Input.}
Let $\P$ be a set of $n$ points in the plane, and let $\G= (\P, \EG)$ be the unit-distance graph of $\P$. Here $\dG$ denotes the shortest path metric of the graph $\G$. The graph $\G$ is assumed to be connected. We are also given a parameter $\eps \in (0,1)$.

\subsubsection{The \WSPD construction algorithm}
\seclab{alg}

We modify the construction of \thmref{WSPD_generic} to use the \WSPD for the Euclidean case for short distances.

\paragraph*{Short distances.}

We compute the $\tfrac{64}{\eps}$-\WSPD $\WR_{<}$ for $\P$, using \thmref{WSPD}, throwing away any pair $\{ A,B\}$ such that $\diamX{A} > 1$ or $\diamX{B} > 1$, where $\diamX{A}$ is the Euclidean diameter of $A$. Let $\WR_1$ be the resulting pair decomposition.

\paragraph*{The $i$\th resolution.}
For $i \in 2, \ldots,2 + \ceil{\log_2 n}$, let $r_i = 3 \cdot 2^{i}$. Let $\Packing_i$ be an $r_i/2$-packing of $\G$, and let
\begin{equation*}
    \WR_i
    =%
    \Set{\bigl.\{ \vorZ{\P}{\Packing_i}{x}, \vorZ{\P}{\Packing_i}{y}\}}{
       x,y  \in \Packing_i
       \text{ and }
       \dGY{x}{y} \in [2r_i/\eps, 16r_i/\eps] }.
\end{equation*}

\paragraph*{The \WSPD decomposition.}

The desired $\tfrac{1}{\eps}$-\WSPD decomposition is the set
$\WR = \cup_i \WR_i$.

\subsection{Analysis}

\begin{claim}
    \clmlab{sep}%
    All the pairs $\{X,Y\} \in \WR$ are $\tfrac{1}{\eps}$-separated in $\G$.  Furthermore, for any distinct $x,y \in \P$, there exists a pair $\{X,Y\} \in \WR$, such that $x \in X$ and $y \in Y$.
\end{claim}
\begin{proof}
    The claim follows by applying \thmref{WSPD} to short pairs, and \thmref{WSPD_generic} to long pairs. The detailed argument, provided next, just spells this out, and the reader might (rightly)  skip it.

    If $\{X,Y\} \in \WR_1$, then $\max(\diamX{X},\diamX{Y}) \leq 1$, and thus $X$ and $Y$ are $\tfrac{1}{\eps}$-separated by \thmref{WSPD}, and they each form a clique in $\G$. We then have that
    \begin{equation*}
        \max(\diamY{\G}{X},\diamY{\G}{Y})
        =
        \max(\diamX{X},\diamX{Y})
        \leq
        \frac{\eps}{64} \deY{X}{Y}
        \leq
        \frac{\eps}{64} \dGY{X}{Y}.
    \end{equation*}
    If $\{X,Y\} \in \WR_i$, for $i > 1$, then the argument in \thmref{WSPD_generic} applies in this case.

    If $\dY{x}{y} \leq 32/\eps$, then the pair $\{X,Y\} \in \WR_1$ that separates $x$ and $y$ has the property that, $\max(\diamX{X},\diamX{Y}) \leq (\eps/64)\dY{x}{y} = (\eps/64)\alpha < 1 $. Namely, $\{X,Y\} \in \WR_1 \subseteq \WR$.  Otherwise, if $\dY{x}{y} > 32/\eps$, then $\dGY{x}{y} \geq \dY{x}{y}$, and the argument used in the proof of \thmref{WSPD_generic} applies.~
\end{proof}

\subsubsection{Size analysis in two dimensions}%

\begin{defn}
    For some $r > 0$, let $\GridX{r}$ denote the axis-aligned uniform \emphw{grid}, with each cell being a translated copy of the square $[0,r]^2$. The number $r$ is the \emphi{sidelength} of $\GridX{r}$. The \emphi{cluster} of a cell $\cellA\in \GridX{r}$ is the square $3\cellA$ (i.e., scaling of $\cellA$ by $3$) formed by the $3 \times 3$ contiguous grid cells centered at $\cellA$.
\end{defn}

\begin{defn}
    \deflab{shadow}%
    For a set $X \subseteq \Re^2$, its \emphi{shadow} is the set
    $\shadowX{X} = \Set{ C \in \GridS}{C \cap X \neq \emptyset}$.
\end{defn}

\begin{figure}[t]
    \centerline{\includegraphics{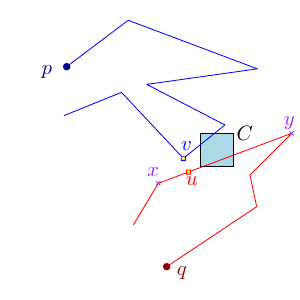}}
    \caption{If two paths intersect the same cell of $\GridX{\tfrac{1}{8}}$, then they can be shortcut.}
    \figlab{obvious}
\end{figure}
\begin{lemma}
    \lemlab{tail_sep}%
    Let $\pi_p$ and $\pi_q$ be two paths in $\G$ with endpoints $p$ and $q$, respectively, such that there is a cell $\cellA \in \shadowX{\pi_p} \cap \shadowX{\pi_q}$.  Then, there is path in $\G$ between $p$ and $q$ of length $\leq \lenX{\pi_p} + \lenX{\pi_q} +1$.
\end{lemma}
\begin{proof}
    The shortest distance between $\pi_p$ and $\pi_q$ is realized by two points $v \in \pi_p$ and $u \in \pi_q$. Observe that $\dY{v}{u} = \deY{\pi_p}{\pi_q} \leq \diamX{\cellA} \leq 1/2$.

    It is easy to verify that one of these points must be a vertex, see \figref{obvious}. So, assume that $v$ is a vertex $\pi_p$, and $u$ lies in the interior of a segment $xy$ of $\pi_q$. By definition $\dY{x}{y} \leq 1$, and thus the distance of $u$ to the nearer endpoint, say $x$, of $xy$ is at most $1/2$. Thus, we have $\dY{v}{x} \leq \dY{v}{u} + \dY{u}{x} \leq 1/2 + 1/2 \leq 1$. Namely, $vx \in \EGX{\G}$, and a path from $p$ to $q$ in $\G$ of the stated length exists.
\end{proof}

\begin{lemma}
    \lemlab{no_friends}%
    Let $R \geq r > 8$, and let $\Packing_r$ be an $r$-packing of $\G$, and consider a point $s \in \Packing_r$. Let $U = U(r,R) = \Set{ y \in \Packing_r}{ 0 < \dGY{s}{y} \leq R}$. Then, $\cardin{U} = O( R^2 / r)$ and $\cardin{\Packing_r}= O( n/r)$ (also, $\cardin{U} \leq \cardin{\Packing_r}$).
\end{lemma}

\begin{proof}
    For any point $x \in U$, let $\pi_x$ be the shortest path from $x$ to $s$ in $\G$.  Let $\Tree$ be the tree formed by $\cup_{x \in U} \pi_x$.  Let $S= \shadowX{\Tree}$ be its \defrefY{shadow}{shadow}, and observe that as $\Tree \subseteq \ballY{x}{R} = \Set{z \in \Re^2}{\dY{s}{z} \leq R}$, we have that $\cardin{S} = O(R^2)$, as each cell of $\GridS$ has area $\Theta(1)$.

    For any $x \in U$, let $\sigma_x$ be the portion of $\pi_x$ starting from $x$ of length $r/2 - 1 > r/3$. By \lemref{tail_sep}, for any $x,y \in U$, we have that $\shadowX{\sigma_x}$ and $\shadowX{\sigma_y}$ are disjoint, as otherwise $\dGY{x}{y} < r$, a contradiction to the packing property.

    Observe that any two non-consecutive vertices of $\pi_x$ must be of length $>1$ from each other. It implies that $\cardin{\shadowX{\sigma_x}} \geq r/c$, for all $x \in U$, for some constant $c$. Thus, $\cardin{U} \leq \cardin{S} / (r/c) =O( R^2 /r )$.

    The second claim follows by considering $R =n$. The total length of the edges of $\Tree$ is $n-1$, as all the edges in $\G$ have length at most $1$ (and thus the shadow of each edge has size $O(1)$), and $\Tree$ is a subtree of $\G$. This implies that $\shadowX{\Tree} = O(n)$, and $\cardin{\Packing_r} = O(n / (r/c))$.
\end{proof}

\begin{lemma}
    \lemlab{resolution}%
    For $i > 1$, we have $\cardin{\WR_i} =O(n/\eps^2)$.
\end{lemma}
\begin{proof}%
    The pairs of $\WR_i$ are induced by pairs of points of $\Packing_i$ that are at most $R_i = O(r_i/\eps)$ apart. \lemref{no_friends} states that $\cardin{\Packing_i} = O(n /r_i)$, and each element of $\Packing_i$ can be matched to at most $d_i = O(R_i^2/r_i) = O(r_i/ \eps^2 )$ other points of $\Packing_i$. Thus, we have $\cardin{\WR_i} \leq \cardin{\Packing_i} d_i = O(n/\eps^2)$.
\end{proof}

\begin{theorem}
    \thmlab{wspd_u_d}%
    For a set $\P$ of $n$ points in the plane, such that the weighted unit-distance graph they induce $\G=(\P,\EG)$ is connected, the above construction results in a $\tfrac{1}{\eps}$-\WSPD of $\G$ of size $O(\eps^{-2} n \log n)$.
\end{theorem}
\begin{proof}
    Summing over the $O( \log n)$ resolutions (the maximum distance in $\G$ is at most $n-1$, so no need to consider higher resolutions), and using \lemref{resolution}, implies the bound on the size of the \WSPD. \clmref{sep} implies the correctness.
\end{proof}

\begin{remark}
    \begin{compactenumA}
        \item \thmref{wspd_u_d} improves over the bound of Gao and Zhang \cite{gz-wspdu-05} by a factor of $1/\eps^2$. The question whether the $O(\log n)$ should be in the bound is still open. The new construction and argument is (arguably) simpler.

        \smallskip%
        \item The computed \WSPD is a cover -- it might cover a single pair of points several times. The classical construction of the \WSPD of points in $\Re^d$ is a disjoint cover of all pairs. Being a cover does not seem to matter for most \WSPD applications.

        \smallskip%
        \item We made no effort to optimize the constants and did not bother to describe an efficient construction algorithm. It seems self-evident that using standard techniques should lead to a near-linear time construction algorithm, as was done by Gao and Zhang \cite{gz-wspdu-05}.
    \end{compactenumA}
\end{remark}

\subsection{Higher dimensions}

\begin{lemma}
    \lemlab{higher_dim}%
    Let $\P$ be a set of $n$ points in $\Re^d$, for $d>2$, and let $\G$ be the weighted unit-distance graph induced on $\P$. Then, for any $\eps \in (0,1)$, the graph $\G$ has a $\tfrac{1}{\eps}$-\WSPD of size $O(n^{2-2/d}/\eps^2)$.
\end{lemma}
\begin{proof}
    The construction is the same as in the 2d case. For $d >2$, the bound of \lemref{no_friends} becomes $\cardin{U(r,R)} = O( R^d/r)$.  Thus, the maximum degree of a point in the graph induced on $\Packing_j$, if we interpret a pair of $\WR_j$ as an edge between points of $\Packing_j$, is at most
    \begin{equation*}
        d_j = \Theta( r_j^{d-1}/\eps^d ).
    \end{equation*}
    Thus, we have $\beta_j = \cardin{\WR_j} = O( d_j \cardin{\Packing_j} ) = O(n r_j^{d-2}/\eps^d)$, as $\cardin{\Packing_j} = O(n /r_j)$.  The quantities $\beta_j$ are geometrically growing with $j$, as $r_j = O(2^j)$. Fortunately, we must have $d_j \leq \cardin{\Packing_j}$. This implies, for the maximum such index $i$, for some constant $c$, the inequality
    \begin{equation*}
        \frac{r_i^{d-1}}{\eps^d }
        \leq
        c \frac{n}{r_i}
        \implies
        r_i^{d}
        \leq
        c
        n \eps^d
        \implies
        r_i = O( \eps n^{1/d}).
    \end{equation*}
    This implies that $\beta_i = \cardin{\WR_i} = O(n^{2-2/d}/\eps^2)$. Thus, as mentioned above, the sum $S = \sum_{j=1}^i \beta_j$ behaves as a geometric sum dominated by the last term (i.e., $S = O( \beta_i)$).

    Finally, observe that $\beta_i$ also dominates the sum $\sum_{j=i+1}^{\log_2 n} \beta_j$, as this is now a decreasing geometric sum with $\cardin{\Packing_j}$ shrinking by a factor of $2$ between consecutive terms (the bound on $\beta_j$ would shrink even faster).
\end{proof}

A nifty (but not too difficult) exercise is to show that the bound of \lemref{higher_dim} is tight, and we leave it as an exercise to the interested reader\footnote{The lazy reader can check out Gao and Zhang \cite{gz-wspdu-05} who provide such an example.}.

\section{Approximating the distortion of a finite map in low dimensions}
\seclab{distortion}

\subsection{Background: Approximating the dilation}

Let $\FMS = (\Ground, \DistChar )$ and
$\FMSB = (\GroundB, \DistCharB )$ be two metric spaces.

\begin{defn}
    \deflab{dilation}%
    Consider a mapping $f: \Ground \to \GroundB$. For $C > 0$, the
    function $f$ is \emphi{$C$-Lipschitz} if
    \begin{equation*}
        \max_{x, y \in \Ground}
        \frac{\dmBY{f(x)}{ f(y)}}
        {\dmY{x}{ y}}
        \leq C.
    \end{equation*}
    The minimum $C$ such that $f$ is $C$-Lipschitz is the
    \emphi{dilation} of $f$ (or the Lipschitz norm of $f$), denoted by
    $\LipX{f}$.
\end{defn}

The dilation can be approximated in near-linear time  if the mapping is
from a point-set that has a small \WSPD decomposition.

\paragraph{Sketch of algorithm \cite{ns-asfeg-00, hm-fcnld-06}.}
Let $P = \{ p_1, \ldots, p_n\}$ be a set of $n$ points in $\Re^d$, and let $f:P \rightarrow (\Ground,\DistChar)$ be a mapping, where $(\Ground,\DistChar)$ is a metric space.  As a reminder, the {dilation} of $f$ is
\begin{math}
    \LipX{f}%
    =%
    \max_{i\neq j}
    \frac{\dmY{f(x_i)}{ f(x_j)}}
    {\dY{x_i}{ x_j}}.
\end{math}
To this end, compute a $\tfrac{8}{\eps}$-\WSPD $\WR$ of $\P$. For each pair $(A, B) \in \WR$, let $\repX{A} \in A$ and $\repX{B} \in B$ be arbitrary representatives of these sets. The algorithm computes and returns the following quantity as the approximate dilation:
\begin{math}
    \ell = \max_{(A,B) \in \WR}
    \frac{\dmY{f(\repX{A})}{ f(\repX{B})}}
    {\dY{\repX{A}}{ \repX{B}}}.
\end{math}

\begin{lemma}[\cite{ns-asfeg-00, hm-fcnld-06}]
    \lemlab{lip}%
    Let $P$ be a set of $n$ points in a metric space, where one can compute a $\tfrac{1}{\eps}$-\WSPD of $\P$ in $O( T(n) )$ time, with $\eps \in (0,1)$ being a prespecified parameter. Let $f:P \rightarrow \FMS$ be a given mapping, where $\FMS$ is some metric space. Then, one can compute in $O(T(n))$ time, a quantity $\ell \in \Re$, such that $(1-\eps)\LipC \leq \ell \leq \LipC$, where $\LipC = \LipX{f}$.  Alternatively, one can output a quantity $\ell'$ such that $\LipC \leq \ell' \leq (1+\eps)\LipC$.
\end{lemma}

\subsection{Approximating the distortion}

\begin{defn}
    For $K \geq 1$, the mapping $f$ is \emphi {$K$-bi-Lipschitz} if
    there exists a constant $\constA > 0$ such that
    \begin{equation}
        \constA K^{-1} \cdot \dmY{ x}{ y}%
        \leq%
        \dmBY{\bigl.f(x)}{ f(y)}%
        \leq%
        \constA \cdot \dmY{x}{y},%
        \eqlab{distortion}%
    \end{equation}
    for all $x,y \in X$.  The least $K$ for which $f$ is $K$-bi-Lipschitz is the \emphi{distortion} of $f$ and is denoted $\Distor(f)$.
\end{defn}

\begin{lemma_no_dot} \textbf{\emph{\cite[p. 356]{m-ldg-02}.}} %
    If $f$ is a bijection, then
    \begin{math}
        \displaystyle \Distor(f) = \LipX{f} \cdot
        \LipX{\smash{f^{-1}}}.
    \end{math}
\end{lemma_no_dot}

\begin{theorem}
    let $P$ be a set of $n$ points in $\Re^d$, $\eps \in (0,1)$ be a parameter, and let $f:P \rightarrow \Re^k$ be a given mapping, where both $d$ and $k$ are constants. Then, one can compute in $O(n \log n + n/\eps^{\max(d,k)})$ time, a quantity $\ell \in \Re$, such that $\Distor(f) \leq \ell \leq (1+\eps) \Distor(f)$.
\end{theorem}
\begin{proof}
    Let $P = \{ p_1, \ldots, p_n\}$, and let $q_i = f(p_i)$, for $i=1, \ldots, n$. Assume that $f$ is one-to-one for now, which implies that $f^{-1}$ is well defined. We $(1+\eps/4)$-approximate $\LipX{f}$ and $\LipX{f^{-1}}$ (from above) using the algorithm of \lemref{lip} on $f$ and $f^{-1}$, respectively.  The product of the two approximations is the desired quantity, by \Eqref{distortion}.

    If $f$ is not one-to-one, then two points of $f(P)$ are the same, and the distortion is infinite. This can be discovered in $O( n )$ time by computing the closest pair in the multisets $P$ and $f(P)$.
\end{proof}

\section{Computing a simple  subcurve via shortcutting}

\subsection{Computing a simple subcurve in the plane}
\seclab{simple_subcurve}

Here, we study the following problem.
\begin{problem}
    Let $\curve \subseteq \Re^2$ be a polygonal curve with $n$ segments, which are self-intersecting. Compute a simple subcurve $\curve' \subseteq \curve$ with the same endpoints as $\curve$.
\end{problem}

Assume $\curve$ has $u$ and $v$ as its endpoints, and it is oriented
from $u$ to $v$.

\paragraph{Algorithm.}

Let \algDistill{}$(\curve)$ denote the recursive algorithm described here. If $\mX{\curve} \leq 2$, then \algDistill simply returns $\curve$ as the desired untangled curve, where $\mX{\curve}$ denotes the number of edges of $\curve$\footnote{%
For simplicity, we ignore ``tedious'' degenerate cases. For example, here, the case that the two edges of $\curve$ are collinear.%
}.

Next \algDistill splits $\curve$ in its median vertex, into two curves $\crA, \crB$, such that $\curve = \crA \concat \crB$, where $\concat$ denotes the concatenation of two curves. Here, $\mX{\crA},\mX{\crB} \leq \ceil{\mX{\curve}/2}$.

Now, the algorithm recursively computes $\craA \leftarrow \algDistill(\crA)$ and $\craB \leftarrow \algDistill(\crB)$. As $\craA$ and $\craB$ are simple, it is natural to return their concatenation as the desired curve, but unfortunately, this merged curve might self-intersect. To avoid this, we compute the first intersection point along $\craA$ with $\craB$ (such an intersection always exists as $\craA$ and $\craB$ share an endpoint).

To this end, we build the ray-shooting data-structure of Hershberger and Suri \cite{hs-pars-95} on $\craB$. Then, for each segment of $\craA{}$, in the order along $\craA$, we query whether this segment intersects $\craB$ by extending the corresponding ray and checking whether the intersection occurs within the given segment. Let $s$ be the first segment of $\craA$ found to intersect $\craB$.

The algorithm computes the intersection points of $\craB$ along $s$ and the first such intersection along $s$ (according to the orientation of $\craA$), denoted by $p$. The subcurve $\craA[u:p]$, the portion of $\crA$ from its start point $u$ to $p$, does not intersect any portion of $\craB$. Thus, if $\craA[u:p] \concat \craB[p:v]$ is the desired simple curve, where $u$ and $v$ are the two endpoints of $\curve$.

\begin{lemma}
    \lemlab{simple_sub_curve}%
    Given a polygonal curve $\curve$ with $n$ segments in the plane, the above algorithm computes a simple subcurve of $\curve$, sharing the same endpoints, in $O(n \log^2 n)$ time.
\end{lemma}
\begin{proof}
    The correctness is hopefully clear. As for running time, we have $T(n) = O(n \log n) + 2T(n/2)= O(n\log^2 n)$.
\end{proof}

\subsection{Shortcutting detours}
\seclab{detour}

The input is a polygonal curve $\curveA$ in $\Re^d$ with $n$ vertices, and a parameter $\alpha > 1$. A \emphi{$\alpha$-detour} is subcurve $\curveA[u,v]$ of $\curveA$ between two vertices $u$ and $v$ of $\curveA$, such that $\lenX{curveA[u,v]} > \alpha \dY{u}{v}$ -- that is, the subcurve is $\alpha$ times longer than shortcutting from $u$ to $v$. We are interested in the problem of repeatedly shortcutting such detours, by replacing in $\curveA$ the subcurve $\curveA[u,v]$ by the direct segment $ uv$, until the resulting curve no longer has $\alpha$-detours.

\begin{observation}
    As a warm-up exercise, observe that the maximum $\beta$, such that a $\beta$-detour in $\curveA$, can be approximated using \lemref{lip}. Indeed, consider the piecewise linear mapping $f: [0,\lenX{curveA}]\rightarrow \Re^d$, that parameterizes the polygonal curve $\curveA$ uniformly.  Let $\P$ be the set of $n$ vertices of $\curveA$, and let $X \subseteq \Re$ be the $n$ values for which $f$ maps to a point of $\P$. Consider the inverse mapping $f^{-1}:\P \rightarrow X$. The \defrefY{dilation}{dilation} of $f^{-1}$ is
    \begin{equation*}
        \max_{x, y \in \Ground}
        \frac{\dmBY{f^{-1}(x)}{ f^{-1}(y)}}
        {\dmY{x}{ y}}
        =%
        \max_{x, y \in \P}
        \frac{\lenX{curveA[x,y]}}
        {\dY{x}{y}},
    \end{equation*}
    which is of course the maximum detour of $\curveA$. This quantity can be $(1-\eps)$-approximated by \lemref{lip} in $O(n \log n + n/\eps^d)$ time.
\end{observation}

For our purposes, we need a slight extension of the above, which provides us with the (approximate) best shortcut between the first and second halves of the curves.

\begin{lemma}
    \lemlab{bi_shortcut}%
    Let $\curveA$ be a polygonal curve in $\Re^d$, with vertices $p_1, \ldots, p_n$ (in this order along $\curveA$). Then, given parameters $\alpha > 1$ and $\eps \in (0,1)$, one can compute, in $O(n \log n + n/\eps^d)$ time, the maximum index $j \leq n/2$, and an index $k > n/2$, such that
    \begin{compactenumi}
        \item $p_j p_k$ is a $\geq (1-\eps)\alpha$-detour, and
        \item for all $s$ and $t$, such that $j < s \leq n/2 < t$, we have $p_s p_t$ is not an $\alpha$-detour.
    \end{compactenumi}
\end{lemma}
\begin{proof}
    One computes the bichromatic \WSPD $\WR$ between $L = \{p_1, \ldots, p_{n/2}\}$, and $R=\{p_{n/2+1}, \ldots, p_n\}$.  One then uses the algorithm of \lemref{lip}, as described in the above observation. For any pair of $\{B,C\} \in \WR$, with $B \subseteq L$ and $C \subset R$, the algorithm also precompute the maximum index $i$, such that $p_i \in B$. Now, we modify the algorithm of \lemref{lip} scans the pairs in $\WR$ and checks for all the pairs whose dilation is above $\alpha$. For all pairs with dilation larger than $(1-\eps)\alpha$, the algorithm returns the one with maximum index $i$.
\end{proof}

\begin{lemma}
    Let $\curveA$ be a polygonal curve in $\Re^d$, and parameters $\alpha > 1$ and $\eps \in (0,1)$. One can compute, in $O((\eps^{-d} + \log n) n \log n)$ time, a subcurve of $\curveA$, formed by a subset of vertices of $\curveA$, such that the new curve is the result of performing a sequence of $(1-\eps)\alpha$-shortcuts on $\curveA$.
\end{lemma}
\begin{proof}
    We imitate the algorithm of \lemref{simple_sub_curve}. Assume $\curveA= p_1 p_2 \cdots p_n$, and let $\curveB_1 = p_1 \cdots p_{n/2}$ and $\curveB_2 = p_{n/2+1} \cdots p_{n}$. We recursively perform (approximate) $\alpha$-shortcuts on $\curveB_1$ and $\curveB_2$. Let $\curveC_1$ and $\curveC_2$ be the resulting subcurves. We now use \lemref{bi_shortcut} to compute the ``last'' (approximate) $\alpha$-shortcut $pq$ for $\curveC_1 \concat \curveC_2$ where $p$ is the last vertex in $\curveC_1$ supporting such a shortcut, and $q$ is the last vertex in $\curveC_2$ with such a shortcut starting at $p$. In the resulting curve after this shortcut, no $\alpha$-detour remains. The algorithm returns it as the desired curve.
\end{proof}

\printbibliography

\end{document}